\documentclass[11pt]{article}
\usepackage[T1]{fontenc}
\usepackage{lmodern}
\usepackage[margin=1in]{geometry}
\usepackage{setspace}
\usepackage{amsmath,amssymb,amsthm}
\usepackage{graphicx}
\usepackage{xcolor}
\usepackage{booktabs}
\usepackage{float}
\usepackage{caption}
\usepackage{hyperref}
\hypersetup{
  colorlinks=true,
  linkcolor=black,
  citecolor=black,
  urlcolor=black
}

\usepackage[authoryear,round]{natbib}

\newtheorem{proposition}{Proposition}
\newtheorem{lemma}{Lemma}
\newtheorem{corollary}{Corollary}

\DeclareMathOperator{\E}{\mathbb E}
\DeclareMathOperator{\Prb}{Pr}
\DeclareMathOperator{\Cov}{Cov}

\title{The Value of Peer Review and the Reward to Reputation\thanks{
Many conversations have informed this paper. I particularly appreciate
thoughtful discussions with Rulof Burger, Helanya Fourie, Willem Fourie, Jesse
Naidoo, Matthew Olckers, Melt van Schoor and Marisa von Fintel, and the many
reader comments on ourlongwalk.com on this and related issues.
I chose the question, assumptions and argument. Anthropic's Claude Code and
OpenAI's Codex assisted with derivations, drafting, literature search,
computational verification and adversarial review. I reviewed every claim and
remain responsible for all errors.}}
\author{Johan Fourie\thanks{Department of Economics, Stellenbosch University.
Email: \href{mailto:johanf@sun.ac.za}{johanf@sun.ac.za}.}}
\date{}

\begin{document}
\maketitle

\begin{abstract}
\noindent Journals cannot referee every paper they receive. When submissions
outrun the time reviewers can give, an editor must decide what to do with the
papers no one reads: reject them, or check them with a faster but less accurate
tool, of the kind artificial intelligence now supplies. A paper the tool passes
carries the same public stamp as one an expert approved, and readers cannot tell
a screened paper from a reviewed one. This paper asks which choice makes the
stamp more valuable. Rejecting the unread
papers keeps the approved pool cleaner, because the faster tool sometimes passes
weak work. Yet rejection can make the stamp worth less: good papers turned away
unread join the pool of unstamped papers and raise the quality against which a
stamped paper is judged. Which effect wins depends on how common good work is,
and the model gives the exact switching point when there is one. The choice also shifts how much a
strong affiliation is worth; and once authors can run the screening tool
themselves, wherever some game it and some do not, the stamp's value is set by
what getting past the screen costs, not by what the screen detects, while expert
judgment allows no such gaming.
\end{abstract}
\vspace{1em}\noindent\textbf{Keywords:} certification; peer review; pedigree;
expert judgment; artificial intelligence

\vspace{0.25em}

\vspace{0.5em}

\newpage

\section{Introduction}

Some research tasks can be written down as rules, and some cannot. Recent work
on AI divides jobs along this line: codifiable work, which follows rules a
machine can now carry out at scale, and messy work, which still needs a
person's judgment \citep{garicano2026messy}. Generative artificial intelligence
keeps widening the set of research tasks that rules can handle. Writing a
plausible manuscript is now largely rule-based \citep{fourie2026writing}; judging
one is not. Expert judgment has not become cheaper, and economists who
study AI in science treat it as the input that stays scarce
\citep{agrawal2026aiscience}. The change has now reached evaluation itself: the
technical checks a referee once ran by hand can be written as rules, and
journals have begun to adopt the tools that run them. This paper works out what
that second change does to the value of certification, the worth of the stamp a
journal puts on a paper. The two changes share one feature: whatever part of
writing can be reduced to rules, anyone can produce, and whatever part of
evaluation can be reduced to rules, anyone can pass.

Journals were crowded before either change. Yearly submissions to the top five
economics journals nearly doubled between 1990 and 2012 while those journals
published fewer articles, so acceptance rates fell from about 15 per cent to
about 6 per cent \citep{card2013nine}. What limits a journal is referee time,
and that time comes at a price which does not clear the market
\citep{engers1998referees}. That time is also shared very unevenly: across the
biomedical literature a small share of researchers carry most of the reviewing
\citep{kovanis2016burden}. AI has now added volume on top of the shortage:
after ChatGPT appeared, submissions to \emph{Organization Science} rose sharply
while measured writing quality fell, both in the submissions and in the reviews
\citep{gartenberg2026more}. The crowding extends beyond journals as well: in weeks
when more working papers are released, each one receives fewer downloads and less
attention, and is later published and cited less \citep{lusher2023congestion}.
A journal that cannot read every submission must decide what to do with the
papers its referees do not reach: reject them unread, or check them with a
faster but less accurate tool, of the kind AI now supplies. This paper asks one
main question. Which choice keeps the journal's stamp more valuable? The natural
guess is that rejecting the unread papers protects the stamp, because it
withholds approval from everything the journal has not checked. That guess can be
wrong. A stamp is valued by comparison: a reader prices a stamped paper against
the papers that have no stamp, so anything that raises the quality of the
unstamped papers lowers what the stamp is worth. Good papers rejected only
because no referee reached them are papers of exactly this kind. Turned away
unread, they join the unstamped papers and raise their average quality, and the
stamp that excluded them is then worth less.

Two further questions follow from the same choice, and the paper answers them
too. What does each rule do to the worth of an author's pedigree, the track
record readers can see? And what is the fast tool's verdict worth once authors
can run the tool themselves? Existing work on certification studies certifiers
who choose what to reveal about the work they have examined. A crowded certifier
faces an earlier choice, what to examine at all, and that work has treated the
screening tool behind the choice as fixed equipment, not as something the
screened authors also hold. Those are the two gaps this paper fills.
\citet{fourie2026messy} applies the codifiable--messy division to certification
and argues that a journal's stamp survives only if the journal commits to
rejecting what it cannot read. The model below shows that this commitment is not
always the right one.

I answer all three questions in a deliberately small model. There is a fixed
batch of papers, each either high or low quality, where high quality means the
paper would pass a careful expert review. The journal's expert reviewers reach
only a random fraction of submissions, say one in four; this intensive review,
when it happens, judges quality without error. The journal chooses what happens
to the papers it does not reach. Under \emph{strict verification} it rejects
them. Under \emph{scalable evaluation} it hands them to an imperfect tool that
passes a high-quality paper with probability $\alpha$, say three in four, and a
low-quality paper with probability $\beta$, say one in five. The two tools
mirror the two kinds of work. Intensive review is the messy input, the judgment
that does not scale; the scalable screen is the part of evaluation that has been
reduced to rules. Both routes end in the same public stamp, so readers see that
a paper is approved but not how. A competitive market pays each paper what it is
worth on average given what the market can see, which is the quality a rational
reader infers from the presence or absence of the stamp. The value of the stamp
is the premium: the gap between the price of a stamped paper and the price of an
unstamped one, and also the most an author would pay to display it. Last, I let
authors fall into two visible pedigree groups, one likelier than the other to
produce high-quality work. Review stays blind to the group, but readers see both
the stamp and which group the author is in.

The precise answer is a single crossing. Suppose the screen lets some bad work
through ($\beta>0$) and still tells good work from bad to some degree
($\alpha>\beta$). Then, as good papers become a larger share of the batch, the
two stamps' values cross at most once, and the crossing point has a formula. When
that crossing lies strictly between zero and one, strict verification gives the
more valuable stamp while good work is scarce, and scalable evaluation gives it
once good work is common. The screen can also be too weak for any crossing to
occur: if it passes too few good papers relative to the bad ones it lets through,
strict verification wins at every share. With the running numbers, one paper in
four reviewed and the screen passing three good papers in four and one bad paper
in five, the crossing is at a good-work share of 0.410. At every share, though,
strict verification leaves the higher average quality among stamped papers. A
cleaner stamp and a more valuable stamp can therefore point to opposite rules.
That gap is the paper's central finding.

The second question, about pedigree, has its own threshold. Here the comparison turns on one summary of the two groups: a
midpoint between how likely each group is to produce good work, defined precisely
in Section 5. When that midpoint lies below the threshold, strict verification
produces the wider gap between the two groups' prices among unstamped papers;
when it lies above, scalable evaluation produces the wider gap. Because the two
thresholds are separate numbers, a pair of groups can fall on one side of the
first and the other side of the second. And when the two groups lie on opposite
sides of the value crossing, they value the stamp more under opposite rules: no
single rule gives both groups their more valuable stamp.

The third question makes the screen's error rates something the model
determines. Section 6 gives authors the same tool the journal
screens with, so an author can pay to fix the rule-checkable flaws the screen
would catch. Fixing them changes what the screen sees, not what the paper
actually contributes, so strict verification is untouched: expert judgment still
reads the quality that no fix reaches. Under scalable evaluation, in every equilibrium
where some flagged authors pay to fix their papers and some do not, the value of
the stamp is pinned down: it equals the fee an author pays to show the stamp plus
the cost of fixing divided by the chance $1-p$ of being sent to the screen. The
stamp is then priced by what getting past the screen costs, not by what the
screen can detect. And if fixing is cheap enough, every paper the screen reads
passes, and the comparison reduces to a simple benchmark: strict verification
gives the more valuable stamp exactly when good papers are fewer than half.

The same comparison explains the pedigree result. An unstamped paper is weaker
evidence of low quality under strict verification, which leaves the overflow
unread, than under scalable evaluation, which screens it. So a missing stamp
lowers an author's price by more under screening, and it does so for both groups.
What this does to the gap between the groups depends on where they start. When
both groups are unlikely to produce good work, the extra bad news lowers both
prices towards zero and narrows the gap between them; when both are likely to, it
widens the gap. A single threshold separates the two cases. This is a statement
about prices, not a claim that either rule changes how much a group's pedigree
tells a reader about quality.

The paper contributes to three literatures: the economics of science, the
theory of certification, and the economics of AI and the division of labour.

The first contribution connects certification to how attention and status are
shared out in science. \citet{hanson2024strain} show that the number of indexed
articles has outrun the researchers available to vet them, and warn that the
resulting strain can blur what a journal's name signals. I take that shortage as
given and work out how the rule for the overflow changes two things at once: what
the signal is worth, and what a strong reputation is worth. Science studies has
long found that how a paper is judged depends in part on who wrote it
\citep{lee2013bias}: an author's prominence sways peer assessments in a
controlled experiment \citep{huber2022nobel}, editors give more weight to reports
from heavily published referees \citep{card2020editors}, and reputation even
shifts how credit is split between teams that reach near-identical results
\citep{hill2025scooped}. Publishing in a handful of top journals strongly
predicts tenure in economics, and the standing of the department where an author
trained predicts where they are hired
\citep{heckman2020tyranny,clauset2015hierarchy}, so what the stamp is worth, and
who can earn it, are questions about careers as much as prices. A second kind of
sorting runs through money: authors at less-resourced institutions and in poorer
countries cluster in lower-fee journals, while better-resourced institutions
publish where the charges are higher
\citep{oldford2023marginalizing,klebel2023apcbarrier}. The sorting studied here
needs no money at all. It survives even when submitting is free and review is
blind. Journals have already used blindness as a lever: the \emph{American
Economic Review} ran a double-blind experiment and found that hiding the author's
name lowered acceptance rates for some groups \citep{blank1991effects}. Even so,
blindness does not settle the matter. When the evaluator never sees who wrote a
paper, the rule for the papers it cannot reach still changes the market reward to
reputation, because it changes what a missing stamp implies. Experiments with AI
evaluations of economics papers show that the tool itself reacts to affiliation,
prominence and gender \citep{pataranutaporn2025peerreview}; the blind-screen case
here isolates a separate, indirect route that works even when the evaluator is
not told who the author is.

The second contribution is to the economics of certification. The closest
predecessor is \citet{farhi2013fear}, where a rejection from a choosy certifier
is itself bad news: every rejected paper was read and found wanting, so rejection
lowers the standing of those turned away and worsens the pool of the rejected.
Rationing by capacity turns this around. A paper rejected only because the
journal ran out of reviewers is not bad news about the paper, because no one read
it. Rejection without reading produces mistakes that carry no stigma, and it can
raise rather than lower the quality of the unstamped pool. That reversal is what
lets a cleaner stamp be worth less, and it cannot happen when every rejection is
informative. The paper adds a second result of the same kind: the tool a
certifier screens with is itself shaped by behaviour. Once the authors being
screened hold that tool, its price shapes the screen's error rates, and where
some game the screen and some do not, the stamp's value is pinned by the cost of
passing it rather than by anything the certifier chose.

A broad literature studies certifiers that choose what to tell the market: what
to disclose, what to charge, how open to be, and how many grades to give
\citep{lizzeri1999information,stahl2017certification,zapechelnyuk2020optimal,
pollrich2024irrelevance}. Work on grading makes a related point: any one grade is
read against the other grades the same scheme hands out
\citep{dranove2010quality,daley2014grades,harbaugh2018coarse}. In another line of
work the standard itself is set by behaviour, either as a norm about what referees
demand that shifts slowly over time \citep{ellison2002evolving}, or as a journal's
best response to its own aims, which can make deliberately lax review worthwhile
\citep{zollman2024journals}. A further line asks how an editor should choose among
referees and combine their reports \citep{bayar2021editorial}. I study a different
margin. When review capacity is spread too thin, what happens to the papers that
get no report, and how does that choice change the price the market puts on a
shared stamp? A separate literature studies authors who game a fixed scoring rule,
and the blurred information that results
\citep{frankel2019muddled,ball2025scoring}; there the question is the score itself
and how the designer should answer the gaming. Here both tools are held fixed,
there is a single public stamp, and I do not model the journal's aims; Section 6
then lets authors game the screen, and that gaming changes which rule gives the
more valuable stamp. What is new here is the cost that rationing imposes on a
stamp's standing: the model sets the two rules against each other for the same
shortage of expert time, and finds the exact share of good work at which the more
valuable stamp switches from one rule to the other.

The third contribution is to the economics of AI and the division of labour.
There the codifiable--messy split sorts tasks: rule-based work is now done by
machines at scale, and messy work is still done by people
\citep{garicano2026messy}. The first of the two changes is already documented.
Employers on a large freelancing platform stopped paying extra for well-written
applications once large language models arrived, because good writing no longer
told strong applicants from weak ones \citep{galdin2025making}, and scientific
output now grows fastest among the authors who used to pay the highest price for
language, those writing outside their first language
\citep{kusumegi2025production}. This paper asks what the same split does to the
institutions that certify the work. Turning an evaluation task into rules changes
more than who does it. A check a journal can run by rule at scale is a check an
author can pass by rule at will, so the very step that lets a journal screen its
overflow lets the screened authors pass the screen. The line between rule-based
and messy work therefore shows up in prices. When some authors game a rule-based
screen and others do not, the value of the stamp it grants is set by what getting
past the screen costs; the premium of a stamp built on judgment answers to no
such cost, because judgment stays messy. In the division-of-labour
literature the messy remainder is the work machines leave to people. Here it also
decides what sets a stamp's value once everyone owns the tools: the value of a
stamp built on rule-based checks comes to depend on how authors respond to the
tool rather than on what the tool detects, while a stamp built on judgment keeps
the value the fundamentals give it. So the line does a second job: it sorts tasks
by who does them, and it sorts stamps by what fixes their value once the tools are
in every author's hands.

All three results inform a live debate about what AI does to the making of
knowledge. Algorithms are expected to reorganise scientific work, not just speed
it up \citep{mullainathan2025algorithms}, and formal models show that when
machine output stands in for human effort, the shared store of knowledge can
shrink even as each person's decisions improve \citep{acemoglu2026collapse}. Two
findings from nearby settings strengthen the concern. Competition lowers measured
quality when the evaluation is rough \citep{hill2025race}, and applicants who can
use generative AI are harder for employers and investors to judge
\citep{cowgill2026does}. Editorial discretion, by contrast, changes what gets
published less than is often assumed \citep{krieger2025gatekeeping}. The debate
has so far said little about certification itself, which is how science turns
private judgment into a public signal. The results below show that a crowded
certifier choosing between the two rules faces a trade-off with no clear winner:
the rule it picks sets both how much a stamp is worth and how much an author's
affiliation is worth to those without one. Recent models of AI and publishing
study richer problems, including submission volume, review costs, capacity and
fees \citep{atasu2026commons,lopezlira2026prompt}, and how to redesign review
when authors and reviewers both use AI \citep{hakobyan2026monitor}.
\citet{gans2026designing} studies the same tools on the other margin, inside
expert review: which referees, if any, should see a journal-commissioned AI
assessment. There too the tool's worth is settled in equilibrium, since a report
that makes each referee more accurate can leave the editor less informed,
because referees who read the same report make the same mistakes. This paper
isolates one piece those problems can build on: the price the market puts on a
shared stamp, and what that price becomes once authors hold the screening tool
themselves.

The rest of the paper runs as follows. Section 2 sets out the model. Section 3
works out the value of a stamp and shows that it equals what an author would pay
for one. Section 4 compares the two rules and states the main result. Section 5
asks whose prices the choice of rule falls on, across pedigree groups. Section 6
puts the screen itself in play, letting authors hold the tool it is built from,
and maps artificial intelligence onto the model's inputs. Section 7 concludes.
Appendix A gathers the proofs; Appendix B extends the model to two noisy screens,
targeted expert review, a continuous measure of quality, and authors who choose
whether to submit.

\section{The model}

There is a fixed batch of completed papers, treated as a unit mass. Each paper
$i$ has a hidden quality $q_i\in\{0,1\}$ that matters for certification. Quality
one means the paper would pass a careful expert review; quality zero means it
would not. Write
\[
 \pi=\Prb(q_i=1)\in(0,1)
\]
for the share of high-quality papers.

The model takes as given whatever a reader can already see. It holds fixed the
visible features of a manuscript, such as its topic or how polished the prose
is. So $q$ is the part of quality that only evaluation can reveal: the soundness
or importance not already obvious to the market that will read it. Among papers
that look the same, the journal cannot rank them before evaluating. Sections
2--4 stay inside one such group of look-alike papers; Section 5 compares two
groups that differ in the author's publicly visible pedigree.

The journal announces one of two rules. Under either rule, a paper gets
intensive review with probability $p\in(0,1)$; intensive review passes a
high-quality paper and rejects a low-quality one. How much review there is,
$p$, is taken as given. A journal could in principle buy more, but referees are
paid little or nothing and give their time for reasons a fee does not easily
change, so the amount of intensive review barely responds to what a journal is
willing to spend \citep{engers1998referees}. Prices are nowhere near the level
at which that would even be tested. Among the 600 highest-ranked economics
journals, only 118 charge a submission fee at all, and the median charge among
those is 125 dollars. Posted fees come to about one per cent of a central
estimate of what an acceptance is worth to an author, and the fee that would
clear the queue, along with the referee pay it would fund, sits well above
anything charged today \citep{fourie2026price}. Review is therefore rationed,
not priced. That companion paper prices entry to the queue; this one prices
what the queue produces, the stamp itself. The interesting case is a $p$ below
the level that would clear the queue, and it is the only case in which the two
rules differ at all: at $p=1$ every paper is examined and the choice
disappears. The rules differ in what they do with the overflow, the papers
intensive review does not reach. Under \emph{strict verification}, a paper that
misses intensive review is rejected. Under \emph{scalable evaluation}, a paper
that misses intensive review is passed with probabilities
\[
 \alpha=\Prb(C\mid q=1,\text{ scalable}),\qquad
 \beta=\Prb(C\mid q=0,\text{ scalable}),
 \qquad 1\geq\alpha\geq\beta\geq0.
\]
Here $\alpha$ is the true-positive rate, the chance a good paper passes the
screen, and $\beta$ is the false-positive rate, the chance a bad one does. The
running example has $p=0.25$, $\alpha=0.75$ and $\beta=0.20$: reviewers reach
one paper in four, and the screen passes three good papers in four and one bad
paper in five. The figures use these numbers throughout.

The journal issues one common public stamp: readers see that a paper is
certified, not which route certified it. This pooling is essential to
everything that follows. If the journal revealed the route, readers would price
intensively reviewed and screened papers separately, and the single-stamp
comparison below would not apply.

Submissions and rejections are private. The market sees whether a paper has the
stamp, not whether it applied and was turned down. A competitive market
pays a paper its expected quality given what the market sees: the chance the
paper is high quality, conditional on the stamp being present or absent. An
author who has earned the stamp may display it by paying a fixed, stamp-only
price $f\geq0$, set outside the model. This price is a measuring device for
willingness to pay; it is not a claim that journals should charge for
acceptance. Submitting a finished paper costs nothing in the baseline. The
order of events is: the journal announces its rule and $f$; authors submit; the
journal evaluates; authors who qualify decide whether to display the stamp; and
the market pays each paper its expected quality.

A little notation finishes the setup. For either rule, write
\[
 x_H=\Prb(C\mid q=1),\qquad x_L=\Prb(C\mid q=0),
\]
for the overall chances that a high- and a low-quality paper end up certified,
and let the certified share be
\[
 m=\pi x_H+(1-\pi)x_L.
\]
Write the two market prices as
\[
 P_C=\E[q\mid C],\qquad P_N=\E[q\mid N],
\]
where $N$ is the pool of papers the public sees without a stamp. The premium is
$\Pi=P_C-P_N$, the gain in reputation from holding the stamp. The uncertified pool
$N$ is terminal: the model has a single certification stage, so a paper the
journal rejects for want of capacity stays in $N$ and earns no other publicly
observed stamp. This is what lets a rejected good paper raise $P_N$, the
mechanism behind the results.

\section{The value of a certificate}

Bayes' rule reduces the value of a stamp to one expression.

\begin{lemma}[Certification premium]\label{lem:premium}
For any $0<m<1$,
\begin{equation}
 \Pi(x_H,x_L)=
 \frac{\pi(1-\pi)(x_H-x_L)}{m(1-m)}.
 \label{eq:premium}
\end{equation}
\end{lemma}

Each part of equation \eqref{eq:premium} does a job. The numerator is the prior
variance of quality, $\pi(1-\pi)$, times how sharply the rule separates good
papers from bad, $x_H-x_L$. The prior variance is only one factor: the same $\pi$
also enters the denominator, so a more uncertain prior does not by itself make
the stamp worth more. What the numerator isolates is the rule's separation, which
raises the premium whenever it improves. The denominator captures how the rule
splits the papers between the two pools, stamped and unstamped. This is where the result comes from. A rule that
certifies more papers leaves fewer in the unstamped pool it is judged against,
and it is that change in the comparison pool, not the cleanliness of the stamp,
that can raise or lower the premium. Every paper the rule does not stamp falls
into that comparison pool.

The premium is not just a statistical gap; it is also what an author would pay
for the stamp. Authors know more about their own papers than the market does.
Let author $i$ privately put the chance her paper is high quality at
$\theta_i\in(0,1)$, with these beliefs averaging out to $\E[\theta_i]=\pi$. Her
chance of ending up eligible for the stamp is
\[
 s(\theta_i)=x_L+(x_H-x_L)\theta_i.
\]
The argument has two steps: first work out the premium assuming every author
submits, then check that submitting is indeed what every author wants. The order
matters: the premium is computed under a guess about behaviour, and the lemma
then shows the guess was right.

\begin{lemma}[Private demand]\label{lem:demand}
Fix a regime and calculate $\Pi$ in the all-submission outcome. For every
$f\in[0,\Pi)$, all authors strictly submit and every eligible author strictly
displays the stamp. The supremum price consistent with strict participation and
display is $\Pi$.
\end{lemma}

An author who earns the stamp gains $P_C-f-P_N=\Pi-f$ from
showing it, and she reaches that point with probability $s(\theta_i)$, so
submitting has expected gain $s(\theta_i)(\Pi-f)>0$. The display price is a
measuring device: it turns the premium into demand for the stamp without our
having to model the journal's own aims. Whichever rule gives the larger premium
is the rule whose stamp authors value more.

\section{Strict versus scalable evaluation}

Section 3 turned the value of the stamp into one formula. That formula now
answers the main question: when a journal cannot read everything, does rejecting
the overflow or screening it keep the stamp more valuable? We compute the two
premiums and compare them. Throughout, a superscript $R$ marks strict
verification, which rejects the overflow, and a superscript $S$ marks scalable
evaluation, which screens it. Under strict verification a good
paper is certified only if intensive review reaches it, and no bad paper is ever
certified, so $x_H^R=p$ and $x_L^R=0$. Hence
\begin{equation}
 \Pi_R=\frac{1-\pi}{1-\pi p}.
 \label{eq:strict}
\end{equation}
For scalable evaluation define
\[
 c=(1-p)\beta,\qquad
 d=p+(1-p)(\alpha-\beta),
\]
where $c$ is a bad paper's overall chance of being certified and $d$ is how much
higher a good paper's chance is. Then $x_L^S=c$, $x_H^S=c+d$, and
\begin{equation}
 \Pi_S=
 \frac{\pi(1-\pi)d}
 {(c+\pi d)(1-c-\pi d)}.
 \label{eq:scale}
\end{equation}

The main result says when each rule gives the more valuable stamp. Its one
condition, interior pools, asks only that both pools, stamped and unstamped,
contain some papers, so that both prices are well defined.

\begin{proposition}[The value of strictness]\label{prop:strictness}
For interior certified and uncertified pools,
\begin{equation}
 \operatorname{sign}(\Pi_R-\Pi_S)
 =\operatorname{sign}K(\pi;p,\alpha,\beta),
 \label{eq:sign}
\end{equation}
where
\begin{equation}
 K(\pi;p,\alpha,\beta)
 =c(1-c)-2cd\pi-d(d-p)\pi^2.
 \label{eq:K}
\end{equation}
If $\beta>0$ and $\alpha>\beta$, $K$ is strictly decreasing in $\pi$,
so the premiums cross at most once. If the positive root is interior, it is
\begin{equation}
 \pi^*=\frac{c(1-c)}
 {\sqrt{c^2d^2+d(d-p)c(1-c)}+cd}.
 \label{eq:root}
\end{equation}
Strict verification has the larger premium below $\pi^*$; scalable evaluation
has the larger premium above it. If $\beta=0<\alpha$, scalable evaluation has
the larger premium for every $\pi\in(0,1)$.
\end{proposition}

In plain terms: which rule wins depends only on the sign of a quadratic in the
share of good work, and that quadratic falls as good work becomes more common.
When its root sits inside the unit interval, strict verification wins while good
work is scarce (low $\pi$), scalable evaluation wins once it is common (high
$\pi$), and the switch happens exactly once, at $\pi^*$. The running example
puts numbers on the switch. The crossing is at $\pi^*=0.410$. When one paper in
five is good, the strict premium is 0.842 and the scalable premium 0.523. At the
crossing the two are equal, at 0.657. When three papers in five are good, the
strict premium is 0.471 and the scalable premium 0.642.

Even when $\beta>0$ and $\alpha>\beta$, the crossing may lie outside the range.
Because $K(0)>0$ and $K$ falls, the positive root is below one exactly when
\begin{equation}
 \alpha[p+(1-p)\alpha]>(1-p)\beta.
 \label{eq:attainability}
\end{equation}
If this fails, strict verification has the larger premium at every share of good
work.

The case with no false positives has a clean interpretation from information
theory. One tool is said to Blackwell-dominate another when it is more
informative for every user who might rely on it, whatever the decision. When
$\beta=0<\alpha$, the screen certifies no bad paper, just as strict verification
does, and certifies more good ones, so it Blackwell-dominates strict review: it
is better on every count. When $\beta>0$ and $\alpha>\beta$, neither tool wins
outright: strict review is better at keeping bad papers out, while the screen is
better at letting good ones in. Since neither is more informative for every
user, the two are not ranked by Blackwell's criterion, and the share of good
work can reverse the ranking.

The result is easiest to see from the two pools. Strict verification never
certifies a bad paper, so the stamped pool is pure, $P_C^R=1$. But a good paper
that misses expert review joins the unstamped pool. When good papers are common,
this raises the unstamped price $P_N^R$ by enough to reduce the premium $\Pi_R$. Scalable evaluation lets some bad papers through, which lowers the
stamped price $P_C^S$, yet it can certify enough good papers that the unstamped
price $P_N^S$ falls by more than $P_C^S$ does. The scalable stamp is then
noisier but worth more.

How the crossing changes with the screen's quality follows the same logic.

\begin{corollary}[Performance of scalable evaluation]\label{cor:performance}
Suppose $\pi^*$ is interior. Holding $p$ and $\beta$ fixed,
$\partial\pi^*/\partial\alpha<0$. Holding $p$ and $\alpha$ fixed,
$\partial\pi^*/\partial\beta>0$.
\end{corollary}

A better hit rate on good papers widens the range of shares over which scalable
evaluation gives the larger premium; more false positives narrows it. The reason
is general: for a fixed $\pi$,
\begin{align}
 \frac{\partial\Pi}{\partial x_H}
 &=\pi(1-\pi)\left\{\frac{x_L}{m^2}
 +\frac{1-x_L}{(1-m)^2}\right\}>0, \label{eq:dxh}\\
 \frac{\partial\Pi}{\partial x_L}
 &=-\pi(1-\pi)\left\{\frac{x_H}{m^2}
 +\frac{1-x_H}{(1-m)^2}\right\}<0. \label{eq:dxl}
\end{align}
Any premium rises when good papers are certified more often and falls when bad
papers are. The crossing changes in the stated directions because
$x_H^S=p+(1-p)\alpha$, $x_L^S=(1-p)\beta$, and $\Pi_R-\Pi_S$ crosses zero from
above.

\begin{figure}[htbp]
 \centering
 \includegraphics[width=\linewidth]{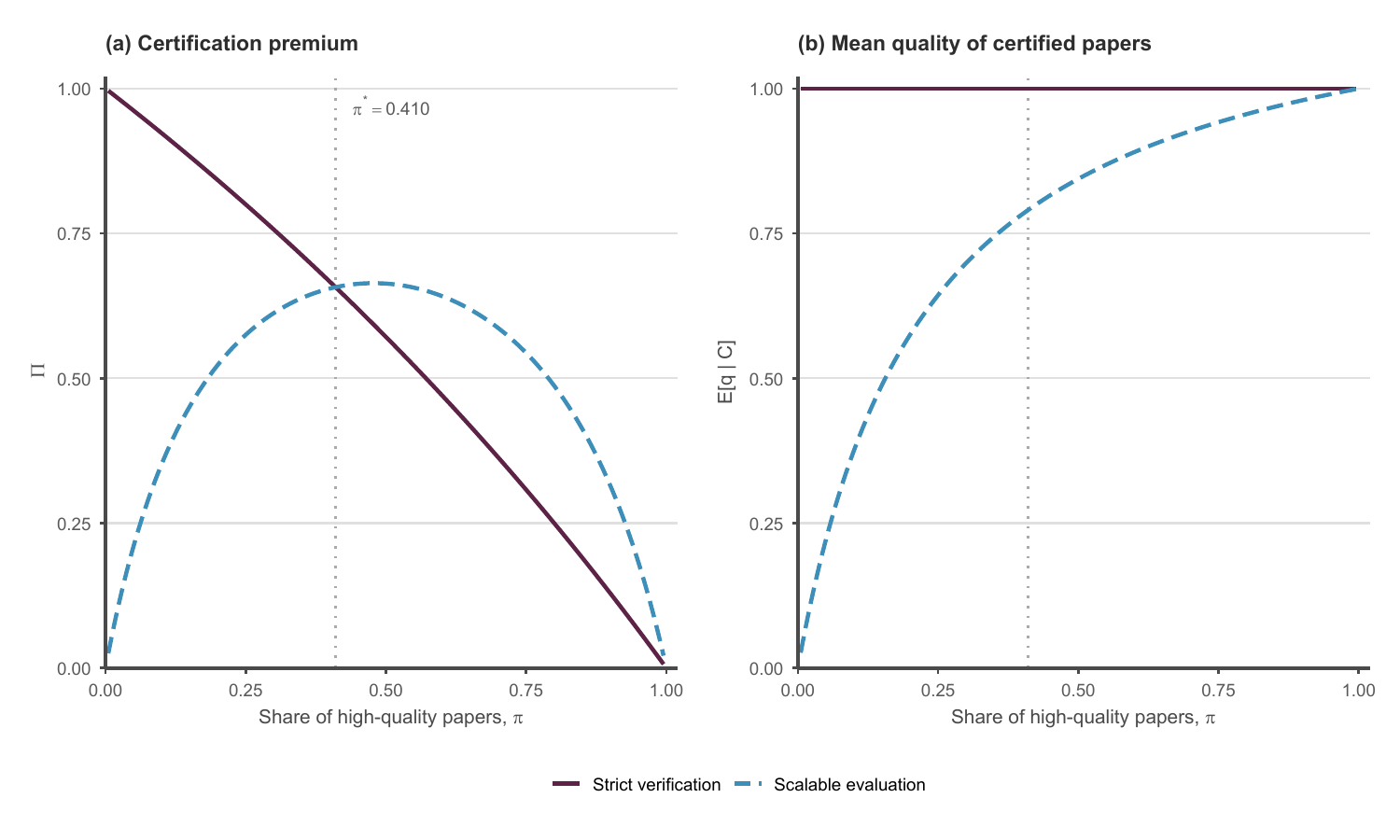}
 \caption{A cleaner stamp need not be a more valuable one. Panel (a) plots the
 premium, the value of holding the stamp, under strict and scalable evaluation
 against the share of high-quality papers. Panel (b) plots the average quality
 of stamped papers. The parameters are $p=0.25$, $\alpha=0.75$ and $\beta=0.20$,
 giving $\pi^*=0.410$. The parameter values are illustrative, not calibrated.}
 \label{fig:wedge}
\end{figure}

Figure~\ref{fig:wedge} shows the gap between how clean a stamp is and how much
it is worth. Strict verification gives the cleaner stamped pool at every share:
$P_C^R=1>P_C^S$. Yet above the crossing scalable evaluation gives the larger
premium. The point is not that the noisier screen yields the cleaner pool; it is
that cleanliness alone does not settle how much a common stamp is worth to the
author who holds it.

One special case makes the threshold easy to see. Suppose the journal certifies
every overflow paper without reading it, so $\alpha=\beta=1$.
Equation~\eqref{eq:K} then reduces to
\begin{equation}
 \operatorname{sign}(\Pi_R-\Pi_S)=\operatorname{sign}(1-2\pi).
 \label{eq:half}
\end{equation}
Strict verification is more valuable exactly when fewer than half of all papers
are good. This benchmark is worth stating on its own. A journal that stamps its
overflow unread is doing the least selective thing it can, and even then
rejecting wins only while good papers are the minority. Where the line sits for
a screen that does sort papers is not fixed by this case. It is fixed by
Corollary~\ref{cor:performance}: a higher hit rate on good papers lowers the
crossing and more false positives raises it, so the threshold can sit on either
side of one half, depending on how good the screen is.

Two clarifications mark the limits of the result. First, $\Pi$ is a difference
in prices, the private gain from holding the stamp. It is not welfare, not the
amount of information produced, not average stamped quality, and not the
certifier's running costs; the model does not treat maximising the premium as
the journal's goal. Nor is the premium the journal's income. Push the display
fee up towards its ceiling $\Pi$; every certified author pays it, so revenue
tends to $m\Pi$. Write $m_R$ and $m_S$ for the certified shares under the two
rules. Scalable evaluation certifies at least as large a share of both types, so
$m_S\geq m_R$, and it sorts the types at least as sharply, so $d\geq p$. By
equation~\eqref{eq:premium} the product $m\Pi$ equals
$\pi(1-\pi)(x_H-x_L)/(1-m)$, which is at least as large under scalable evaluation
on both counts: the larger certified share shrinks the denominator, and the
sharper sorting raises the numerator. So its potential revenue is at least
strict review's at every share. The crossing is about the value of one
certificate, not total fees. Second, the proposition compares the two rules for
a fixed batch of completed papers. Once submitting costs something and the net
premium is fixed, strict review draws a stronger set of authors, because weaker
authors are less willing to try. Near the crossing this change in who submits
can reverse the ranking. Appendix~\ref{app:extensions} works this out.

\section{Pedigree and the incidence of certification}

Section 4 answered the value question for one pool of look-alike papers. The
same model answers a second question: which authors gain or lose from the rule.
Let the market now see something about the author as well as the journal's stamp.
Suppose authors fall into two visible pedigree groups $b\in\{B,0\}$, where a
share $\pi_0$ of one group and a larger share $\pi_B$ of the other produce
high-quality work, $0<\pi_0<\pi_B<1$. Pedigree might be the
author's institution, membership of a research network, or an established name.
Posting in a prominent working paper series is the clearest example, and
crowding within such a series measurably cuts the attention each paper gets
\citep{lusher2023congestion}. The difference between the groups is a statistical
link with quality, not a taste for or against anyone. The question is one of
incidence, in the sense the word has in the study of taxes: on whose prices does
the choice of rule fall? Review is blind: how often intensive review
reaches a paper, and the screen's error rates, do not depend on the group $b$.
Some journals have adopted and tested such blindness \citep{blank1991effects},
and the results below show what it does not achieve. Because the market does see
the group, every earlier price formula applies group by group, with $\pi$
replaced by the group's own share $\pi_g$.

What review does to an unstamped paper is captured by a likelihood ratio $\ell$,
how likely a good paper is to go unstamped relative to a bad one,
\begin{equation}
 \ell=\frac{\Prb(N\mid q=1)}{\Prb(N\mid q=0)}
 =\frac{1-x_H}{1-x_L}.
 \label{eq:ell}
\end{equation}
For a group with share $\pi$, the market price of an unstamped paper is
\begin{equation}
 U(\pi;\ell)=\E[q\mid N;\pi]
 =\frac{\ell\pi}{1-\pi+\ell\pi}.
 \label{eq:unstamped}
\end{equation}
The two rules give
\begin{equation}
 \ell_R=1-p,
 \qquad
 \ell_S=\frac{(1-p)(1-\alpha)}{1-(1-p)\beta}.
 \label{eq:ells}
\end{equation}
When $0\leq\beta<\alpha<1$, $0<\ell_S<\ell_R<1$. A smaller likelihood ratio
means that having no stamp points more strongly to low quality, so a missing
stamp is worse news under scalable evaluation than under strict verification.
The reason is that the screen looks at every overflow paper, while strict
verification leaves the papers it cannot reach unexamined.

Define the pedigree gap among unstamped papers under rule $r\in\{R,S\}$ as the
difference in the two groups' unstamped prices,
\[
 W_r=U(\pi_B;\ell_r)-U(\pi_0;\ell_r).
\]
Because the two groups need not be close together, the comparison leans on one
summary of the pair: the logit midpoint, the group share whose odds are the
geometric mean of the two groups' odds, so that its log-odds sit exactly halfway
between the two groups' log-odds,
\[
 \widetilde\pi=
 \operatorname{logit}^{-1}\left\{
 \frac{\operatorname{logit}\pi_0+\operatorname{logit}\pi_B}{2}
 \right\}.
\]

\begin{proposition}[Pedigree and uncertified work]\label{prop:pedigree}
Under blind review with $0<\pi_0<\pi_B<1$ and
$0\leq\beta<\alpha<1$,
\begin{equation}
 \operatorname{sign}(W_R-W_S)
 =\operatorname{sign}(\widehat\pi-\widetilde\pi),
 \label{eq:pedigreesign}
\end{equation}
where
\begin{equation}
 \widehat\pi=
 \frac{1}{1+\sqrt{\ell_R\ell_S}}
 =\frac{1}{1+(1-p)
 \sqrt{(1-\alpha)/[1-(1-p)\beta]}}\in(1/2,1).
 \label{eq:pedigreethreshold}
\end{equation}
Strict verification produces the larger pedigree-related price gap among
uncertified papers exactly when $\widetilde\pi<\widehat\pi$. Moreover,
\begin{equation}
 \frac{\partial\widehat\pi}{\partial\alpha}>0,
 \qquad
 \frac{\partial\widehat\pi}{\partial\beta}<0,
 \qquad
 \frac{\partial\widehat\pi}{\partial p}>0.
 \label{eq:pedigreecomparative}
\end{equation}
\end{proposition}

The proposition stops just short of $\alpha=1$, and that boundary case runs in
strict verification's favour. At $\alpha=1$ the screen passes every good
overflow paper, so under scalable evaluation an unstamped paper is known to be
bad, both groups' unstamped prices are zero, and strict verification opens the
larger gap for every pair. This is the limit of the stated result, since
$\widehat\pi\to1$ as $\alpha\to1$.

The threshold lies above one half for every allowed set of parameters, which
fixes the direction of the result whenever the two groups' quality shares are
low. Suppose the two shares add to less than one. Then the product of their odds
is below one, the logit midpoint is below one half, and one half is below the
threshold. So strict verification produces the wider gap. The condition is that
the two group-specific shares sum below one, not that good papers are a minority
of the whole population, which would also turn on how many authors sit in each
group. A crowded certifier whose two groups meet it therefore cannot avoid the
following combination: rejecting what it cannot examine keeps its stamped pool
clean and, at the same time, raises what an author's affiliation is worth among
the papers it turned away. Scalable evaluation reverses the comparison only when
both groups are strong enough to raise their midpoint above the threshold.

The result compares market rewards, not how much pedigree actually tells a
reader about quality. In odds terms pedigree is worth exactly the same under
both rules, because a blind stamp multiplies both groups' starting odds by the
same likelihood ratio:
\begin{equation}
 \frac{\operatorname{odds}(q=1\mid N,B)}
 {\operatorname{odds}(q=1\mid N,0)}
 =\frac{\pi_B/(1-\pi_B)}{\pi_0/(1-\pi_0)}
 \label{eq:oddsinvariance}
\end{equation}
under either rule. The gap in log-odds does not change. The rule also leaves each
group's average price untouched: by the law of iterated expectations a group's
stamped and unstamped prices average back to its prior $\pi_g$. What the rule
moves is the spread between those two prices, and with it the gap $W_r$ among
unstamped papers, not any group's expected reward. Prices are a different
matter, because the market pays expected quality in levels, and going from odds
to levels is not a straight-line step. A smaller $\ell$ makes an unstamped paper
stronger evidence of low quality. For two weak groups, that evidence is close to
decisive for both, so both prices are near zero and nearly equal, and the gap
between them is small. For two strong groups, the same evidence makes each price
react sharply to the group's exact share, so a small difference in share produces
a large difference in price, and the gap is wide. The formula behind this is the
slope of the unstamped price as the group's share rises, $\ell/[1-(1-\ell)\pi]^2$:
a smaller $\ell$ makes it flatter where shares are low and steeper where they are
high. The threshold $\widehat\pi$ is the share at which the two slopes are equal.

\begin{figure}[htbp]
 \centering
 \includegraphics[width=\linewidth]{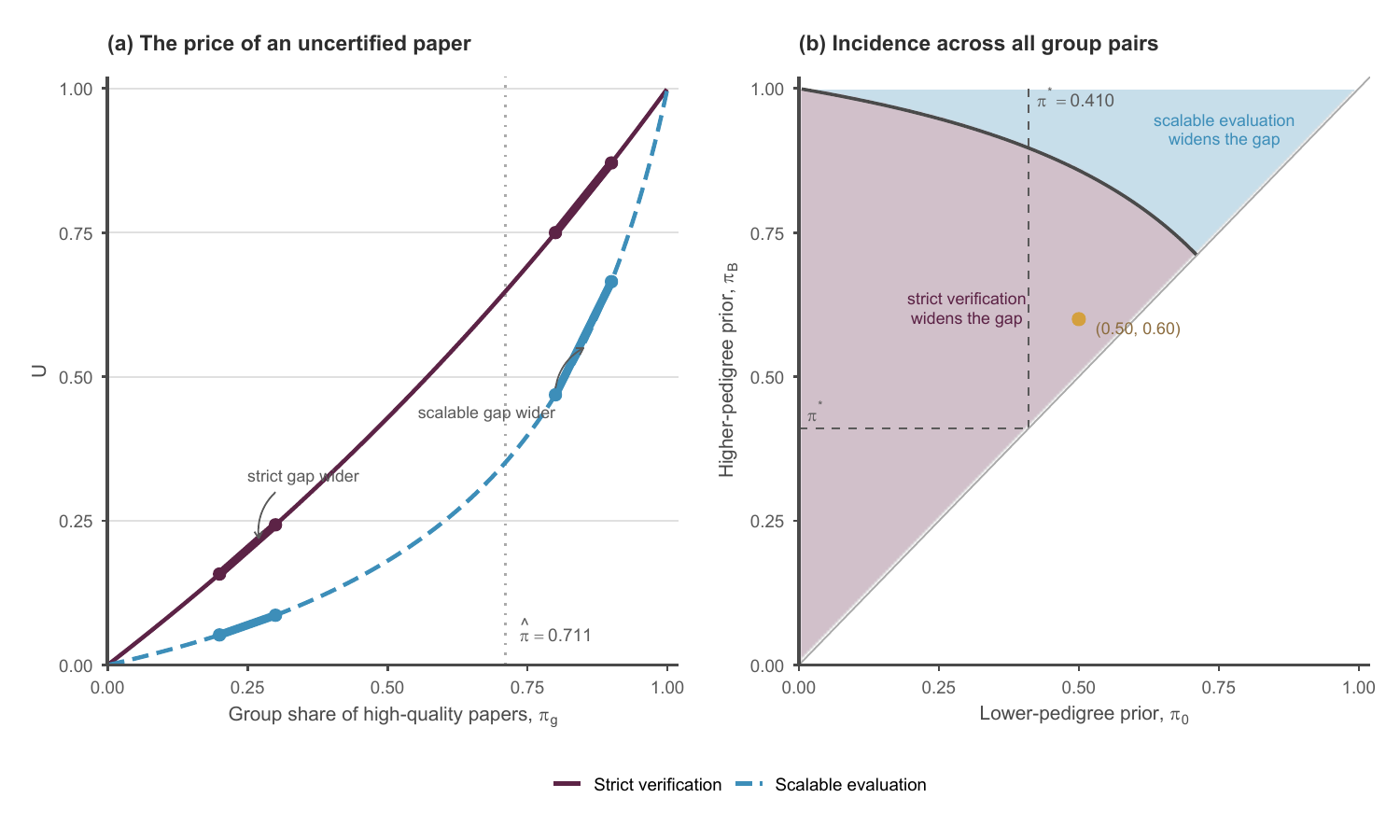}
 \caption{On whose prices the choice of rule falls. Panel (a) plots the price of
 an unstamped paper against a group's share of good work under the two rules, at
 the parameters of Figure~\ref{fig:wedge}. The marked pairs are groups at
 $(0.20,0.30)$ and $(0.80,0.90)$; the line joining each pair is the gap between
 their prices. Panel (b) shades, for every pair of group shares with
 $\pi_B>\pi_0$, the rule that produces the larger gap. The solid curve is the
 boundary $\widetilde\pi=\widehat\pi$, the dashed lines mark the premium
 crossing $\pi^*$ of Proposition~\ref{prop:strictness}, and the gold point is
 the pair $(0.50,0.60)$ discussed below. The parameter values are illustrative,
 not calibrated.}
 \label{fig:incidence}
\end{figure}

Figure~\ref{fig:incidence} draws both halves of the result. In panel (a) the
scalable price line lies below the strict line everywhere, flatter where shares
are low and steeper where they are high, and their order reverses at
$\widehat\pi=0.711$. The groups at shares 0.20 and 0.30 are separated by 0.085
under strict verification and by 0.034 under scalable evaluation. The groups at
0.80 and 0.90 are separated by 0.121 under strict verification and by 0.196
under scalable evaluation. Panel (b) covers every possible pair of group shares.
Over most of the square strict verification produces the larger gap, which follows
from $\widehat\pi>1/2$.

Pedigree is also priced differently among stamped papers. Under strict review
$\E[q\mid C,b]=1$ for both groups, so among stamped papers the two groups fetch
the same price and the gap is zero. Under scalable evaluation the gap is
positive exactly when $\beta>0$. The zero gap under strict review comes from
taking intensive review to be perfect; it need not survive if intensive review
is itself noisy.

The premium crossing of Section~4 has one more consequence for who pays what.

\begin{corollary}[Regime disagreement]\label{cor:groupdisagreement}
Suppose the crossing $\pi^*$ in Proposition~\ref{prop:strictness} is interior.
If $\pi_0<\pi^*<\pi_B$, then
\begin{equation}
 \Pi_R(\pi_0)>\Pi_S(\pi_0),
 \qquad
 \Pi_S(\pi_B)>\Pi_R(\pi_B).
 \label{eq:groupdisagreement}
\end{equation}
\end{corollary}

An author from the weaker group will pay more for the stamp under strict
verification; an author from the stronger group will pay more under scalable
evaluation. This ranks what each group would pay to hold the stamp, the
willingness to pay of Lemma~\ref{lem:demand}; it is not a ranking of which rule
each group would choose overall. An author's expected gain from submitting is her
premium times her chance of being certified, and strict verification, by
rejecting the overflow, lowers that chance. A weaker author can therefore value
the stamp more under strict review yet be better off overall under scalable
evaluation, which certifies her more often. At the running parameters, an author
in a group where good work is one paper in five, who privately rates her own
chance at that same one in five, expects 0.042 from submitting under strict
review but 0.148 under scalable evaluation, though the strict premium of 0.842 is
larger than the scalable one of 0.523. If a single display fee applies to both
groups, every author still submits as long as the fee is below the smaller of the
two groups' premiums.

The two thresholds answer different questions, and the rankings they induce need
not agree. The premium crossing $\pi^*$ is compared with one group's share and
ranks what the stamp is worth to that group. The pedigree threshold $\widehat\pi$
is compared with the logit midpoint of a pair of groups and ranks the gap between
their prices among unstamped papers. The two rankings agree at the two ends.
Suppose both groups' shares lie below $\pi^*$ and the pair's midpoint lies below
$\widehat\pi$. Then strict verification gives both groups the more valuable
stamp and also produces the wider gap between them among unstamped papers. For an
author in the weaker group this is the high-stakes rule: it raises what she
gains by being stamped and what she loses by not being. When instead both shares
lie above $\pi^*$ and the midpoint lies above $\widehat\pi$, scalable evaluation
gives both groups the more valuable stamp and the wider gap, so it is then the
high-stakes rule and strict verification shrinks both differences. In
between, the two rankings can point different ways.

Neither condition implies the other. The gold pair in
Figure~\ref{fig:incidence}, with shares 0.50 and 0.60, has its logit midpoint at
0.55, below $\widehat\pi=0.711$, so strict verification produces the wider pedigree
gap; yet both shares exceed $\pi^*=0.410$, so scalable evaluation gives the more
valuable stamp to both groups. Which threshold is larger depends on the
parameters, and nothing in the model forces the two comparisons to agree. All of
this is about prices; none of it is a ranking of welfare.

\section{The codified screen in equilibrium}

The scalable screen is no longer hypothetical. In August 2026, Refine, an AI
referee built by economists to run scalable technical checks
\citep{golub2026refine}, announced partnerships with two leading publishers
in economics. Where such a tool sits in the process decides which part of the
model it touches. Run after acceptance, as a final technical audit, it only
checks papers expert judgment has already chosen; it arrives after the stamp is
given and changes nothing above. Run on the overflow, it is exactly the scalable
screen of this model. This section asks the third question: what is the screen's
verdict worth once the authors it screens can run the tool themselves?

Let the screen read a rule-checkable mark on each manuscript, the result of the
checks it can run by rule. A high-quality paper has a passing mark with
probability $\alpha$ and a low-quality paper with probability $\beta$, and the
screen certifies an unreviewed paper exactly when its mark passes. These are the
same rates as before, now with a story behind them. Suppose the tool that runs
the screen is also in the authors' hands before they submit: an author can read
the paper's mark from the tool's report and pay a cost $\kappa>0$ to fix the
rule-checkable flaws the screen would catch. Fixing changes the mark, not the
quality, because quality is the messy part no rule reads. The market sees
neither the mark nor the fixing, only the stamp, and prices using how authors
actually behave. Two assumptions keep this honest: the screen is not already
useless (it does not pass everything), and the mark is a feature of the
manuscript alone, learned only from the report and, once quality is held fixed,
unrelated to anything else the author knows.

Under strict verification, fixing is worthless. An unreviewed paper is rejected
whatever its mark, and a reviewed paper is judged on quality, which fixing does
not change. Under scalable evaluation, a fixed paper is certified in the overflow
instead of rejected, so every author whose mark fails gains the same amount from
fixing: the premium minus the fee $f$ to display the stamp, when that difference
is positive and nothing otherwise, times the chance $1-p$ of being sent to the
screen. That the gain is the same for good and bad papers is the crux: being
indifferent about fixing does not sort good papers from bad.

\begin{proposition}[Erosion of the codified screen]\label{prop:erosion}
Under scalable evaluation with repair cost $\kappa>0$:
\begin{enumerate}
\item[(i)] The strict-verification premium is unaffected at every $\kappa$.
\item[(ii)] In every equilibrium in which a positive mass of failing-mark
papers is repaired and a positive mass is not, the scalable premium equals
\[
 \Pi_S=f+\frac{\kappa}{1-p}.
\]
\item[(iii)] Let $\Pi^{\min}$ be the premium when repair reaches every
failing-mark low-quality paper and no high-quality one, given in closed form
in the appendix. If $\kappa<(1-p)\bigl(\Pi^{\min}-f\bigr)$, the unique
equilibrium is full repair: the screen passes everything it reads, scalable
evaluation certifies the overflow unfiltered, and
$\operatorname{sign}(\Pi_R-\Pi_S)=\operatorname{sign}(1-2\pi)$.
\end{enumerate}
\end{proposition}

Part (ii) pins the premium, in every partial-fixing equilibrium, at the display
fee plus the cost of fixing divided by $1-p$. Fixing pays off only in the
fraction $1-p$ of papers that reach the screen, so its cost is spread over that
fraction and enters divided by $1-p$. The logic is a no-arbitrage argument. If
the stamp were worth more than that sum, every flagged author would fix; if it
were worth less, none would; a premium at which some flagged authors fix and
some do not must equal it exactly. Which authors fix is not pinned down. When
authors have some private read on their own quality, fixing can concentrate on
exactly the papers the screen would have caught, and the screen's real error
rates can converge until it separates nothing at all, while the premium stays
exactly at the pinned level. Equilibrium disciplines the price of the
stamp, not who ends up behind it.

Across the repair cost the equilibrium shifts from one regime to another: a high
cost leaves the screen intact and no one repairs, an intermediate cost gives the
partial-repair equilibria whose premium part (ii) pins, and a low cost brings on
full repair, which is the unique equilibrium when the screen's floor is positive
and one equilibrium among several otherwise. The pinning identity rests on every flagged author facing the same
repair cost, which is what stops repair from sorting good papers from bad. If the
cost fell unevenly by quality, repair would carry information about quality and
the identity need not hold exactly; whether the screen then separates more or
less depends on which papers are cheaper to fix.

Part (iii) is the limit the introduction promised. When fixing is cheap enough,
the screen certifies everything it reads, and the choice between the two rules
reduces to the clean benchmark of Section~4, equation \eqref{eq:half}: strict
verification gives the more valuable stamp exactly when fewer than half the
papers are good. Full repair is an equilibrium whenever repairing pays for itself
at full participation, that is when $\kappa$ is at most $(1-p)$ times the
full-repair premium net of the fee. For full repair to be the \emph{only}
equilibrium, the screen's worst-case premium $\Pi^{\min}$ must be positive, which
at a display fee of zero requires coverage $p$ to exceed $(1-p)(1-\alpha)$, the
share of good papers that sit in the overflow and fail the screen; at the running
parameters, one quarter exceeds three sixteenths. Where that floor is not
positive, full repair can still occur, but alongside other equilibria.

Erosion does not always lower the premium. When good work is common, the intact
screen does harm of its own: it wrongly rejects good overflow papers and so
raises the quality of the unstamped pool, and dissolving the screen stops this.
Whether full fixing lowers the premium is therefore a single-crossing question:
the premium when the overflow is passed unfiltered is $\pi/[1-p(1-\pi)]$, and when
the quadratic recorded in the appendix changes sign inside the unit interval,
this unfiltered premium is larger than the intact screen's above a third
threshold $\pi^{u}$. At the running parameters $\pi^{u}=0.580$. Below it the
intact screen has the higher premium, so erosion lowers the premium: at a
good-work share of 0.20 it falls from 0.523 to 0.250. Above it, for the reason
just given, erosion raises the premium: at 0.60 it rises from 0.642 to 0.667. The paper's
three thresholds run in the order $\pi^*=0.410$, then $\pi^u=0.580$, then
$\widehat\pi=0.711$ at these parameters, and the order depends on the parameters.

The same tool can instead be sold per manuscript, so that an author buys the
report and the fix together without first seeing the mark. Then, when the stamp
is worth displaying and the screen still separates at all, the authors willing to
pay most are the ones least sure of their work, because their papers are the
likeliest to fail the screen. So the authors who buy are those least sure,
everyone below a confidence cutoff, and their buying raises the false-positive
rate by at least as much as the true-positive rate, so any interior cutoff
erodes the screen's separation. This is
the pattern experiments already find: applicants who can use generative AI are
harder for employers and investors to judge \citep{cowgill2026does}. None of this
denies that fixing genuinely improves a manuscript on the things a rule can
check; the model says nothing about that value. What full fixing removes is
information: a pass mark everyone can obtain separates no one, while the messy
quality underneath is unchanged.

\citet{fourie2026messy} reaches the one-half benchmark a different way, through a
flood of submissions while the journal stamps its overflow unread. The erosion
limit reaches the same boundary for a fixed batch of papers, by dissolving the
filter instead. When the publisher partnerships were announced, Garicano noted
that authors will run the tool before submitting, so that in equilibrium papers
``will have passed by it beforehand and the AI referee will find nothing new''
\citep{garicano2026judgement}. In the model this is exactly right in the
full-fixing equilibrium; short of that, authors' indifference pins the premium
instead.

The other links from AI to the model give a map, not a forecast, because none
of them has a sign that holds unconditionally. Let AI capability be $a$, the
volume of completed papers $n(a)$, and intensive-review capacity $A(a)$. Under
random rationing,
\[
 p(a)=\min\{1,A(a)/n(a)\},
\]
where at $p=1$ every paper gets intensive review, the two rules coincide, and
there is nothing to compare. AI may also change the share of good work $\pi(a)$
and the screen's performance $(\alpha(a),\beta(a))$. The dividing line for policy
is simply
\[
 K\bigl(\pi(a);p(a),\alpha(a),\beta(a)\bigr)=0.
\]
This setup allows both a flood of AI-written submissions, which lowers $p$, and
AI-assisted review, which raises $A$ or improves $(\alpha,\beta)$. The
introduction's split between rule-based and messy work organises these links. A
tool raises $\alpha$ without raising $\beta$ when it turns more of the evaluation
task into rules, as verification-first proposals aim to
\citep{you2026verification}, and the proposition above gives the terms on which
that gain survives once authors hold the same tool. It raises $A$ when it saves
referees time on the rule-based parts of their work. Even that route has a
design margin this model holds fixed: \citet{gans2026designing} shows that
giving every referee the same AI report can reduce what the editor learns from
the referees together, because their remaining errors become correlated, so the
report may be best kept by the editor or given to one referee only. That result
concerns the papers experts do reach; the erosion above concerns the papers
they do not. The messy core, the
judgment itself, stays the scarce input. The one term with no settled sign is
the share of good work $\pi(a)$: the same tools that make submissions harder to
screen can also let authors finish projects whose ideas were sound but whose
execution was held back, and these two effects have opposite signs. Whether a
given tool helps or hurts is an empirical question; the model points to the
quantities that settle the answer rather than guessing it.

Blindness is itself a choice the journal makes. If the screen sees who the author
is and reacts to affiliation or prominence, as recent AI experiments suggest
\citep{pataranutaporn2025peerreview}, its error rates have to be tracked group by
group, $(\alpha_g,\beta_g)$. Section 5 instead isolates the indirect price effect
that survives when the journal hides pedigree from both expert and scalable
review. Direct bias in the evaluator would be a further channel on top of it.

\section{Conclusion}

Does a journal that cannot review everything protect the value of its stamp by
rejecting what it cannot examine? Not always. Strict verification keeps the
stamped pool clean at every share of good work. Yet when the screen sorts well
enough, and once good work is common enough, the strict stamp is worth less than
the noisier scalable one. The line between the two cases is a single crossing in
the share of good work, which the model gives in closed form, and the model also
says when the screen is too weak for that line to be reached at all. Nor is the
screen a fixed alternative. Because scalable evaluation certifies by rule,
authors who hold the same rule can pass it, and in every equilibrium where some
pay to do so and some do not, the scalable premium equals the display fee plus
the cost of fixing, made larger because fixing pays off only in the papers that
reach the screen. When fixing is cheap enough, the screen certifies
everything it reads and the comparison reduces to the clean benchmark: strict
verification gives the more valuable stamp exactly when fewer than half the
papers are good.

The three thresholds are separate quantities. At the illustrative parameters
used throughout, review reaching one paper in four and a screen passing three
good papers in four and one bad in five, they fall in a strict order: the
premium crossing at a good-work share of 0.410, the erosion threshold at 0.580,
the pedigree threshold at 0.711. A journal can then give the more valuable stamp
under scalable evaluation, see that premium cut by authors fixing their papers,
and still be the rule that narrows the pedigree gap, all at once. The order
depends on the parameters; the point is that the three comparisons need not line
up.

The rule also sets how the market rewards pedigree. Even when review is blind,
the two rules create different unstamped pools. A second closed-form threshold
says when strict verification produces the wider pedigree gap among unstamped papers
and when scalable evaluation does. And two groups on opposite sides of the
premium crossing value the stamp more under opposite rules. So rationing expert
attention falls on some authors more than others, even when no one is treated
differently on purpose.

The results revise three beliefs. A journal that keeps its stamp scarce by
rejecting what it cannot read is not thereby protecting the stamp's value;
wherever the crossing lies inside the range, there are shares of good work at
which it is lowering it. Blind review does not settle what pedigree is worth; the
rule for the overflow does. And a journal that adopts a well-tested screen has
not bought a tool with fixed error rates; it has entered an equilibrium in which
its own authors' choices set those rates, and in which, whenever some authors
game the screen and others do not, its verdict is worth exactly what getting past
it costs. The third belief matters most for the debate about adopting
these tools: the accuracy a screening tool shows in trials need not be the
accuracy, or the value, it keeps once everyone can use it.

Taken together, the results separate three properties of a certifier. The purity
of its stamped pool, the chance a stamped paper is good, depends on that pool
alone. The value of its stamp compares the stamped and unstamped pools. The
market reward to pedigree compares visible author groups within those pools.
A review rule changes all three, and they need not rank institutions the same
way. The scarce input throughout is the messy one. Expert judgment is the
resource the journal cannot scale, and it is also the one form of evaluation an
author cannot pass in advance: when some authors game a rule-based screen and
others do not, its stamp is priced by the cost of the gaming, and no gaming cost
exists for judgment. The three thresholds say how institutions divide the
consequences of that scarcity.

That separation is why the argument reaches beyond journals, to any institution
built the same way: one pooled public stamp, a comparison pool of those it turns
away, and expert attention too scarce to reach everyone. Grant panels, hiring
committees, accreditation bodies and rating agencies often share that structure,
and each rations the expert attention behind the stamp. Each has
the same two options for the applications it cannot reach, and each is judged in
public by how clean its approvals look. The model says that this standard can
misrank them, because the pool an evaluator turns away is priced too; that any
check an institution can run by rule, applicants can in time pass by rule, so the
price of a certificate built on rule-based checks can equal the cost of passing
them rather than the quality of what passes; and that a certifier can examine
every applicant blind to who they are and still change what an applicant's
affiliation is worth. Rationing expert attention is therefore not a back-office
decision that leaves the market untouched. It sets what a credential is worth,
and what an author's affiliation is worth to those who do not hold one.

The results hold for a fixed batch of completed papers, one public stamp, and
competitive pricing on what the market can infer. The pedigree result also needs
review that is blind to the author's group. If the journal reveals which route a
paper took, readers price the two routes separately and the pooling the results
depend on is gone. If submitting is costly, the two rules attract different sets
of authors, and near the crossing that change in who submits can reverse the
ranking. If the evaluator reacts to pedigree directly, group-specific error rates
add a separate channel. The results also treat the uncertified pool as terminal:
a paper rejected for capacity stays in the comparison pool rather than moving on
to certification elsewhere, and where it is recertified at another venue the
outside pool that prices the stamp is different and the institutional reach of
the results narrows. None of these price comparisons is a ranking of welfare
or fairness.

What would settle where real institutions stand relative to the crossing is a
paired-assessment study. The same manuscripts would be run through both a
scalable screen and an independent expert review; treating that expert review as
the gold standard, the pairs pin down $(\pi_g,\alpha,\beta)$ within each pedigree
group, and editors' own records fix the coverage $p$. RePEc could supply a
population of economics papers to draw from. Publication outcomes on their own
will not do, because a paper rejected for lack of capacity need not reveal how an
expert would have judged it. Such a study would pin down $(\pi_g,\alpha,\beta)$
and the coverage $p$, but not how rejected papers come to populate the market's
comparison pool, the engine of the main result; tracking where capacity-rejected
papers end up is a separate empirical step. Until that evidence exists, the
model's message is
conditional but exact: to judge a crowded certifier, one has to price both the
pools its rule creates and the groups that fill them.

\newpage
\appendix

\section{Proofs}\label{app:proofs}

\begin{proof}[Proof of Lemma~\ref{lem:premium}]
Bayes' rule gives
\[
 P_C=\frac{\pi x_H}{m},\qquad
 P_N=\frac{\pi(1-x_H)}{1-m}.
\]
Subtracting and using $m=\pi x_H+(1-\pi)x_L$ yields
equation~\eqref{eq:premium}.
\end{proof}

\begin{proof}[Proof of Lemma~\ref{lem:demand}]
A stamp-eligible author earns $P_C-f$ by displaying and $P_N$ by declining, so
she strictly displays when $f<\Pi$. Not submitting also yields $P_N$ because
applications are private. Submission followed by optimal display has expected
gain $s(\theta_i)(\Pi-f)>0$. Under strict review
$s(\theta_i)=p\theta_i>0$; under scalable evaluation it is weakly larger. At
$f=\Pi$, the relevant choices are weak, making $\Pi$ the supremum strictly
sustainable price.
\end{proof}

\begin{proof}[Proof of Proposition~\ref{prop:strictness}]
Substitute equations~\eqref{eq:strict} and \eqref{eq:scale}, multiply by their
positive denominators, and collect terms. The sign of the difference is the
sign of $K$. If $\alpha>\beta$, then $d>p$, and
\[
 K_\pi=-2cd-2d(d-p)\pi<0
\]
when $\beta>0$. Solving the quadratic and rationalising its positive root gives
equation~\eqref{eq:root}. If
$\beta=0$, then $c=0$ and $K=-d(d-p)\pi^2<0$. For attainability, observe that
$K(0)=c(1-c)>0$ and
\[
 K(1)=(1-p)\{(1-p)\beta-\alpha[p+(1-p)\alpha]\}.
\]
The decreasing function therefore crosses inside $(0,1)$ exactly when
equation~\eqref{eq:attainability} holds.
\end{proof}

\begin{proof}[Proof of Corollary~\ref{cor:performance}]
Equations~\eqref{eq:dxh} and \eqref{eq:dxl} follow by differentiating the two
posterior means. Increasing $\alpha$ raises $x_H^S$ and therefore $\Pi_S$ at
every prior; increasing $\beta$ raises $x_L^S$ and lowers $\Pi_S$. At an
interior crossing, Proposition~\ref{prop:strictness} implies that
$\Pi_R-\Pi_S$ crosses zero from above. The implicit-function theorem gives the
two signs.
\end{proof}

\begin{proof}[Proof of Proposition~\ref{prop:pedigree}]
Blind review implies $\Prb(C\mid q,b)=\Prb(C\mid q)$, so Bayes' rule applies
separately within each pedigree cell. Equation~\eqref{eq:unstamped} gives
\[
 U(\pi_B;\ell)-U(\pi_0;\ell)
 =\frac{\ell(\pi_B-\pi_0)}
 {[1-(1-\ell)\pi_0][1-(1-\ell)\pi_B]}.
\]
Multiplying $W_R-W_S$ by its positive denominators leaves the sign of
\[
 (\ell_R-\ell_S)
 \{(1-\pi_0)(1-\pi_B)-\ell_R\ell_S\pi_0\pi_B\}.
\]
Because $\ell_R>\ell_S$, the expression is positive exactly when
\[
 \sqrt{\frac{\pi_0}{1-\pi_0}
             \frac{\pi_B}{1-\pi_B}}
 <\frac{1}{\sqrt{\ell_R\ell_S}}.
\]
Applying the increasing logistic transformation gives
equations~\eqref{eq:pedigreesign} and \eqref{eq:pedigreethreshold}.
The denominator of $\widehat\pi$ is one plus a positive term. Increasing
$\alpha$ reduces that term, while increasing $\beta$ raises it. For $p$, set $u=1-p$.
The same term becomes
$u\sqrt{(1-\alpha)/(1-u\beta)}$, which is strictly increasing in $u$ and
therefore strictly decreasing in $p$. This proves
equation~\eqref{eq:pedigreecomparative}.
\end{proof}

\begin{proof}[Proof of Corollary~\ref{cor:groupdisagreement}]
Proposition~\ref{prop:strictness} states that $\Pi_R(\pi)>\Pi_S(\pi)$ below
an interior $\pi^*$ and that the inequality reverses above it. Evaluating the
two group-specific premiums at $\pi_0<\pi^*<\pi_B$ gives
equation~\eqref{eq:groupdisagreement}.
\end{proof}

\begin{proof}[Proof of Proposition~\ref{prop:erosion}]
Describe repair behaviour by the pattern $(r_H,r_L)\in[0,1]^2$, the shares of
failing-mark high- and low-quality papers repaired. The certification rates
are
\[
 x_H=p+(1-p)\bigl[\alpha+(1-\alpha)r_H\bigr],\qquad
 x_L=(1-p)\bigl[\beta+(1-\beta)r_L\bigr],
\]
and Lemma~\ref{lem:premium} gives the premium $\Pi(r_H,r_L)$.
(i) Strict verification rejects every unreviewed paper at every mark, and a
reviewed paper is judged on its quality, which repair does not change.
(ii) In such an equilibrium some failing-mark author weakly prefers
repairing and some weakly prefers not. The gain from repair,
$(1-p)\max\{\Pi-f,0\}$, is common to all failing-mark authors, so both weak
preferences hold for every one of them: indifference,
$\kappa=(1-p)(\Pi-f)$, which rearranges to the identity displayed in part
(ii). Reallocations of repair across indifferent authors that hold the
premium at this level preserve equilibrium, so which authors repair is
not identified, while the premium is.
(iii) By equations~\eqref{eq:dxh} and \eqref{eq:dxl}, $\Pi$ is strictly
increasing in $x_H$ and strictly decreasing in $x_L$, so its minimum over
patterns is at $(r_H,r_L)=(0,1)$:
\[
 \Pi^{\min}=\frac{\pi(1-\pi)\,d_w}{m_w(1-m_w)},\qquad
 d_w=p-(1-p)(1-\alpha),\qquad
 m_w=\pi\bigl[p+(1-p)\alpha\bigr]+(1-\pi)(1-p).
\]
Suppose $\kappa<(1-p)(\Pi^{\min}-f)$. A no-repair equilibrium requires
$\kappa\geq(1-p)[\Pi(0,0)-f]$ and a partial-repair equilibrium requires
$\kappa=(1-p)[\Pi(r_H,r_L)-f]$; both right-hand sides are at least
$(1-p)(\Pi^{\min}-f)$, which exceeds $\kappa$, a contradiction. Full repair
is an equilibrium because
$(1-p)[\Pi(1,1)-f]\geq(1-p)(\Pi^{\min}-f)>\kappa$. At full repair every
overflow paper is certified, which is the transparent benchmark of
Section~4, equation~\eqref{eq:half}.
\end{proof}

The threshold $\pi^u$ of Section~6 compares the intact screen with
unfiltered overflow. Applying the argument of
Proposition~\ref{prop:strictness} to this pair of technologies and dividing
the resulting expression by the positive constant $p$ leaves
the sign of
\[
 J(\pi)=d(1-p)-c(1-c)+2d(c-1+p)\,\pi+d(d-p)\,\pi^2,
\]
whose derivative $2d(1-p)[(\alpha-\beta)\pi-(1-\beta)]$ is strictly negative
on $(0,1)$ when the screen is not trivial. Full repair lowers the premium
exactly while $J(\pi)>0$. When $J(0)>0>J(1)$, the threshold $\pi^u$ is the
unique root of $J$ in $(0,1)$, the smaller root of the quadratic when
$\alpha>\beta$; at $p=0.25$, $\alpha=0.75$ and $\beta=0.20$ it equals
$0.580$.

\section{Robustness and boundaries}\label{app:extensions}

\subsection{Two noisy classifiers}

Perfect intensive review is useful for the single-crossing result, but the
posterior representation does not require it. For regime $j\in\{R,S\}$ write
$c_j=x_{Lj}$ and $d_j=x_{Hj}-x_{Lj}>0$. Then
\begin{align}
 \operatorname{sign}(\Pi_R-\Pi_S)=\operatorname{sign}\{&
 d_Rc_S(1-c_S)-d_Sc_R(1-c_R) \notag\\
 &+2d_Rd_S(c_R-c_S)\pi
 +d_Rd_S(d_R-d_S)\pi^2\}. \label{eq:arbitrary}
\end{align}
To obtain this expression, apply Lemma~\ref{lem:premium} with
$m_j=c_j+\pi d_j$ and expand
$d_Rm_S(1-m_S)-d_Sm_R(1-m_R)$. Two arbitrary binary classifiers can generate
at most two isolated interior crossings unless their premiums coincide. The
assumption that intensive review is perfect sharpens this to at most one
isolated crossing.

\subsection{Targeted intensive review}

Suppose a pre-review signal assigns intensive review with probabilities $r_H$
and $r_L$ for high- and low-quality papers. Strict verification has
$(x_H^R,x_L^R)=(r_H,0)$, whereas scalable evaluation has
\[
 x_H^S=r_H+(1-r_H)\alpha,\qquad
 x_L^S=(1-r_L)\beta.
\]
Define $c=(1-r_L)\beta$ and $d=x_H^S-x_L^S$. Direct substitution gives
\begin{equation}
 \operatorname{sign}(\Pi_R-\Pi_S)=\operatorname{sign}
 \{c(1-c)-2cd\pi+d(r_H-d)\pi^2\}. \label{eq:targeted}
\end{equation}
If
\[
 (1-r_H)\alpha\geq(1-r_L)\beta,
\]
then $d\geq r_H$ and the expression in braces is weakly decreasing in $\pi$.
Unless the two regimes coincide, the one-crossing result survives. If the
condition fails, a global one-crossing claim is false in general.

\subsection{Continuous quality}

Let $q$ have any integrable distribution and let
$x(q)=\Prb(C\mid q)$, with $m=\E[x(q)]\in(0,1)$. Then
\begin{equation}
 \E[q\mid C]-\E[q\mid N]
 =\frac{\Cov(q,x(q))}{m(1-m)}. \label{eq:covariance}
\end{equation}
Indeed, the difference between
$\E[qx(q)]/m$ and
$\{\E[q]-\E[qx(q)]\}/(1-m)$ equals the right-hand side. The binary model is
therefore the pass--fail version of a general covariance identity.

\subsection{Selective entry}

Let submitting cost $k>0$ and private assessments have distribution $F$. For a
candidate cutoff $t$, authors with $\theta\geq t$ submit. Define
\[
 M_C(t)=\int_t^1s(\theta)\,dF(\theta),\qquad
 Q_C(t)=x_H\int_t^1\theta\,dF(\theta),\qquad
 \bar\theta=\int_0^1\theta\,dF(\theta).
\]
The premium is
\begin{equation}
 \Pi(t)=\frac{Q_C(t)}{M_C(t)}
 -\frac{\bar\theta-Q_C(t)}{1-M_C(t)}, \label{eq:entrypremium}
\end{equation}
and an interior cutoff satisfies
\begin{equation}
 s(t)[\Pi(t)-f]=k. \label{eq:entrycutoff}
\end{equation}
With fixed intensive-review capacity $A$, coverage becomes
$p(t)=\min\{1,A/[1-F(t)]\}$. Holding the net premium fixed, strict
verification gives low-assessment authors a smaller chance of certification
and therefore induces stronger positive selection. This composition effect
coexists with the baseline outside-pool effect and may reverse its ranking near
the zero-cost crossing. The main
proposition should consequently be read as the direct certification effect for
a fixed completed-paper cohort, not as an unconditional ranking of optimal
journal policies.

\section*{Data availability}

This article is a theoretical contribution and uses no empirical data. The
computational scripts that verify the analytical results and generate the two
figures recompute every reported quantity from the closed-form expressions in
the text. They are openly available in a public repository; a persistent
identifier will be added before publication.

\bibliographystyle{aea}
\bibliography{references}

\end{document}